\documentclass[final]{siamltex}

\usepackage{amsmath,amssymb,hyperref,latexsym}
\usepackage{xcolor}
\hypersetup{colorlinks}
\usepackage{graphicx}
\usepackage{epsf}
\usepackage{subfigure}
\usepackage{amsfonts}
\usepackage{graphics}
\usepackage{epsfig}
\usepackage{float}

\usepackage{enumerate}

\usepackage[ruled,vlined,linesnumbered]{algorithm2e}

\usepackage{tikz-cd}

 \newtheorem{example}[theorem]{Example}

 \numberwithin{equation}{section}

\newcommand{\beq}{\begin{equation}}
\newcommand{\eeq}{\end{equation}}

\newcommand{\bR}{\mathbb{R}}

\newcommand{\R}{\bR}
\newcommand{\Rn}{\bR^n}

\newcommand{\Rmn}{\bR^{m\times n}}

\newcommand{\argmin}{\mathop{\mathrm{argmin}}}

\newcommand{\sgn}{\mathrm{sgn}}

\newcommand{\norm}[1]{\left\Vert #1\right\Vert}

\newcommand{\ip}[2]{\left\langle #1,\, #2\right\rangle}

\newcommand{\map}[3]{#1:\, #2\rightarrow #3}

\renewcommand{\supp}[1]{\mathrm{supp}(#1)}

\newcommand{\eps}{\epsilon}

\usepackage{xcolor}
\definecolor{newcolor}{rgb}{.8,.349,.1}

\newcommand{\RN}[1]{%
  \textup{\uppercase\expandafter{\romannumeral#1}}%
}

\title    {On the Global Minimizers of Real Robust Phase Retrieval with Sparse Noise
}

\author{%
  Aleksandr Aravkin\footnote{Department of Applied Mathematics, University of Washington}%
  \and James Burke\footnote{Department of Mathematics, University of Washington} 
   \and Daiwei He\footnote{Department of Mathematics, University of Washington}%
  }\date{\today}

\begin{document}

\maketitle
\begin{abstract}
We study a class of real robust phase retrieval problems 
under a Gaussian assumption on the coding matrix when the received
signal is sparsely corrupted by noise. The goal is to 
establish conditions on the sparsity
under which the
input vector can be exactly recovered. The recovery problem is
formulated as the minimization of the $\ell_1$ norm of the residual.
%
%
The main contribution is a robust phase retrieval counterpart 
to the seminal paper by Candes and Tao on compressed sensing ($\ell_1$ regression)
[\emph{Decoding by linear programming}. IEEE Transactions on Information Theory, 51(12):4203–4215, 2005].
Our analysis depends on a key new property on the coding matrix which we call the 
\emph{Absolute Range Property} (ARP). 
This property is an analogue to the Null Space Property (NSP) in compressed sensing. 
When the residuals are computed using squared magnitudes, 
we show that
ARP follows from a standard Restricted Isometry Property (RIP). However, when the residuals are computed using absolute magnitudes,
a new and very different kind of RIP or growth property is required.
We conclude by showing that the robust phase retrieval objectives are sharp with respect to their minimizers with high probability.
\end{abstract}

%

\section{Introduction}
Phase retrieval has been widely studied in machine learning, signal processing and optimization. The goal of phase retrieval is to recover a signal $x$ provided the observations of the amplitude of its linear measurements:
\begin{equation}\label{ic_101}
	|\ip{a_i}{x}| = b_i, \quad  1 \leq i \leq m
\end{equation}
where $a_i \in \mathbb{C}^n$ or $\mathbb{R}^n$, $b_i \in \mathbb{R}$ are observations, 
and  $x$ is an unknown variable we wish to recover (e.g. see \cite{luke2002optical}). A well studied form of the phase retrieval problem is
\begin{equation}\label{ic_102}
|\ip{a_i}{x}|^2 = b_i, \quad 1\leq i \leq m,
\end{equation}
where $b_i$ now represent the squared magnitudes of the observations. 
It is shown in \cite{pardalos1991quadratic} that the phase retrieval problem is NP-hard. Recent work on the phase retrieval problem \cite{charisopoulos2019low, duchi2017solving, davis2017nonsmooth, candes2015phase, candes2014solving} focuses on the real phase retrieval problem where it is assumed that $a_i \in \mathbb{R}^n$ for each $i=1, 2, ..., m$. This is the line of inquiry we follow. 
In the following discussion 
the $m$ rows of the matrix $A \in \Rmn$ are the vectors $a_i \in \Rn$. 


The two most popular approaches to the real phase retrieval problem are through semidefinite programming relaxations \cite{balan2006signal, candes2014solving, candes2013phaselift, chen2015solving, demanet2014stable,  li2013sparse, waldspurger2015phase} and convex-composite optimization \cite{burke1985descent, davis2017nonsmooth, duchi2017solving}. These approaches formulate real phase retrieval problem as an optimization problem of the form
\begin{equation}\label{new_intro_1}
	\min_x \rho(|Ax|^2 - b),
\end{equation} 
where $\rho$ is chosen to be either the $\ell_1$ or the square of $\ell_2$ norm, and,  for any vector $z \in \mathbb{R}^m$, $|z|$ and $z^2 $ are vectors in $\mathbb{R}^m$ whose components are the absolute value and squares of those in $z$. 
The objective in \eqref{new_intro_1} is a composition of a convex and a smooth function, and is called convex-composite. 
This structure plays a key role in both optimality conditions and algorithm development for \eqref{new_intro_1} \cite{burke1985descent}.

 In the noiseless case, when there exists a vector $x_* \in \mathbb{R}^n$ such that $|Ax_*|^2 = b$ (or, $|Ax_*| = b$), a gradient based method called Wirtinger Flow (WF) 
 was introduced by \cite{candes2015phase} to solve the smooth problem
 \[\min_x \norm{|Ax|^2 - b}_2^2.\] 
 WF  admits a linear convergence rate when properly initialized. Further work along this line includes the Truncated Wirtinger Flow (TWF), e.g., see \cite{chen2015solving}. Truncated Wirtinger Flow requires $m \geq Cn$ measurements as opposed to the $m \geq Cn\log n$ measurements in WF to obtain a linear rate. A similar approach using sub-gradient is used to minimize $\min_x \norm{|Ax| - b} _2^2$ in \cite{wang2018solving} for the noiseless case. 

\noindent
{\bf Contributions.}
In this paper we address two forms of the \emph{robust} phase retrieval problem, where the optimization objective takes the form
\begin{equation}\label{new_intro_2}
\min_x f_p(x) :=  \norm{|Ax|^p - b}_1\quad\mbox{for $p=1,2$},
\end{equation}
and it is assumed that the matrix $A$ satisfies the following Gaussian assumption:
\[
\mathrm{G:}\qquad\mbox{The entries of $A$ are i.i.d. standard Gaussians $N(0, 1)$.}
\]
Our goal is to establish a robust phase retrieval counterpart 
to the seminal paper by Candes and Tao on compressed sensing ($\ell_1$ regression) \cite{candes2005decoding}. 

Compressed sensing problems \cite{donoho2006compressed} take the form
\begin{equation}\label{cs_00}
	\min_y \norm{y}_1 \text{ such that } \Phi y = c, 
\end{equation}
where $\Phi \in \mathbb{R}^{n \times N}, y \in \mathbb{R}^N, c \in \mathbb{R}^n$.
This problem is known to be equivalent to the $\ell_1$ linear  regression problem
\begin{equation}\label{ll1}
	\min_x \norm{Ax - b}_1	,
\end{equation}
where $\Phi b = -c$ and $A \in \mathbb{R}^{N \times (N - n)}$ (e.g., 
the columns of $A$ form basis of $\text{Null}(\Phi)$).
In \cite{candes2005decoding} it is shown that there is a universal constant 
$s \in (0, 1)$ such that, 
under suitable conditions on $A$ (e.g., Assumption G), 
if $x^*$ satisfies $\norm{Ax_* - b}_0 \leq sm$, then $x^*$ is the unique solution
to \eqref{ll1}, with high probability. We prove similar exact recovery results for the two robust phase retrieval problems \eqref{new_intro_2}. In particular, we show that 
$\{x_*,-x_*\}=\argmin f_p$ with high probability, when $m \geq 2n-1$ (Theorem \ref{g_thm}).
In this situation, the solution set to $\min f_p$ and the $\ell_0$ phase retrieval problem coincide, that is, 
\begin{equation}\label{ctri_l00}
	\{x_*, -x_*\} = \argmin_{x} \norm{|Ax|^p - b}_0.
\end{equation}
Thus, the $\ell_0$ phase retrieval problem can be solved by the $\ell_1$ phase retrieval problem $\min f_p$, when there exists an $x_*$ with sufficiently sparse noise.

A key underlying structural requirement used by \cite{candes2005decoding} is the Restricted Isometry Property (RIP). 
We also make use of an RIP property in the $p=2$ case. However, 
in the $p=1$ case a new property, which we call the p-Absolute Growth Property (p-AGP) (see Definition \ref{def:agp}), is required. 
When $p=2$, RIP implies 2-AGP. The p-AGP holds under 
Assumption G, with high probability (see Lemmas \ref{ic_lm_0} and  \ref{lm:new6}).
A second key property, which mimics the so-called Null Space Property (NSP) in compressed sensing \cite{cohen2009compressed,daubechies2010iteratively, donoho2001uncertainty, gribonval2002sparse}, 
is also introduced. We call this the
p-Absolute Range Property (p-ARP) (see Definition \ref{def:ARP}), and
show that p-AGP implies p-ARP under Assumption G with high probablility.
In \cite{candes2008restricted}, it is shown that, for problem \eqref{cs_00}, if $\Phi$ satisfies RIP with parameter $\delta_{2s} < \sqrt{2}-1$, then $\Phi$ satisfies NSP of order $s$. 
Correspondingly, we show that the p-AGP implies the p-ARP
with high probability under Assumption G.
(see Lemmas \ref{ic_lm_main} and \ref{ic_lm_main_2}).

There are separate classes of methods for solving \eqref{new_intro_2} for $p=2$ and $p=1$.
When $p = 1$, one can apply a smoothing method to the absolute value function \cite{irls,luke2002optical}, 
or use other relaxation techniques that preserve the nonsmooth objective but introduce auxiliary variables~\cite{zheng2018relax}.
When $p = 2$, the solution methods typically exploit the convex-composite structure of the objective $f_2$. 
These methods rely on two key conditions on the function $f_2$: weak convexity (i.e., $f + \frac{\rho}{2} \norm{\cdot}^2$ is convex for some $\rho > 0$) and sharpness (i.e., $f(x) - \min f \geq c \cdot \text{dist}(x, \mathcal{X})$ for some $c > 0$ where $\mathcal{X}$ is the set of minimizers of $f$). Under these two properties, Duchi and Ruan \cite{duchi2017solving}, Drusvyatskiy, Davis and Paquette \cite{davis2017nonsmooth} and Charisopoulos, et al.\cite{charisopoulos2019low} establish convergence and iteration complexity results for prox-linear and subgradient algorithms.
Recently \cite{zhang2016provable} and \cite{chen2017robust} considered gradient-based methods for the problem $\min_x f_2(x)$ when the noise is $sm$ sparse for some $s < 1$. 
To establish locally linear convergence of their algorithms
the authors of
\cite{zhang2016provable} require that the measurements 
satisfy $m \geq cn\log n$ for $c > 0$, while 
the authors of \cite{chen2017robust} require that $s<c/\log m$ for some $c>0$. The results in \cite{duchi2017solving} and \cite{charisopoulos2019low} require 
$m \geq cn$ for some $c>0$ and for some $s \in [0, \frac{1}{2})$ sufficiently small.
 

Conditions for the weak convexity of $f_2$ follow from results in \cite{duchi2017solving, davis2017nonsmooth} under assumptions weaker than Assumption G. 
In the noiseless case, the sharpness of $f_2$ also follows from 
	results in \cite{duchi2017solving, davis2017nonsmooth}. In the noisy case, sharpness is established in \cite{duchi2017solving, davis2017nonsmooth} under same assumptions on the sparsity of the noise. 
	
	 We establish sharpness for both $f_1$ and $f_2$ under Assumption G uniformly for all possible supports of the sparse noise. Our result for $p=2$ case has a similar flavor to those in \cite{duchi2017solving, charisopoulos2019low}, but more closely 
	 parallels the result of Candes and Tao in the compressed sensing case. When $p=1$, our result has no precedence in the literature and requires a new approach. 
	 The function $f_1$ is not weakly convex since it is not even subdifferentially regular \cite{luke2002optical}.
	 
 This paper is organized as follows. 
 In section 2, we introduce the
 new properties p-ARP and p-AGP and provide a detailed 
 description of how our program of proof parallels the program 
 used in compressed sensing.
 In Section 3, we show that if $A$ satisfies p-ARP and the residual 
 $|\,|Ax_*|^p-b|$ is sufficiently sparse, then $\{\pm x_*\}\subset\argmin f_p$ with
 equality under Assumption G.
 In section 4, we show that Assumption G implies that p-AGP 
 implies p-ARP with high probability. 
 In the last section we show that $f_p$ is sharp with respect to $\argmin f_p$, with high probability.
	 	
\subsection{Notation}Lower case letters (i.e. $x$, $y$) denote vectors, while $x_i$ denotes the $i$th component of the $x$. 
$c_0, c_1, c_2, \tilde{c_0}, \tilde{c_1}, C$ denote universal constants. $\norm{x}$, $\norm{x}_1$ denote the Euclidean and $\ell_1$ norms of vector x, while 
$\norm{x}_0$ denotes the $\ell_0$ `norm' $ |\{i | x_i \neq 0\}|$. For a matrix $X$, $\norm{X}_F$ denotes the Frobenius and $\norm{X}$ denotes the $\ell_2$ operator norm. When $x = (x_i)_{1\leq i \leq n}$ is a vector, $|x|:= (|x_i|)_{1\leq i \leq n}$ and $x^p:= (x_i^p)_{1\leq i \leq m}$. 
For a vector $v \in \mathbb{R}^m$, and $T \in [m]:= \{1,2,...,m\}$, $v_T$ is defined to be a vector in $\mathbb{R}^m$ where the $i$th entry is $v_i$ if $i \in T$ and $0$ else where.  
$\supp{x} := \{i | x_i \neq 0\}$. We say a vector $x$ is $L$ sparse if $\norm{x}_0 := |\supp{x}| \leq L$. 

\section{The Roadmap}
Recall from the compressed sensing literature \cite{cohen2009compressed,daubechies2010iteratively} that a matrix $\Phi \in \Rmn$ satisfies Null Space Property (NSP) of order $L$ at $\psi \in (0,1)$ if 
\begin{equation}\label{nsp}
\norm{y_T}_1 \leq \psi \norm{y_{T^c}}_1\quad
\forall\,y \in \text{Null}(\Phi)\text{ and }|T| \leq L.
\end{equation}
It is shown in \cite{donoho2001uncertainty, gribonval2002sparse} that
every $L$-sparse signal $y_* \in \mathbb{R}^m$ is the unique minimizer of the compressed sensing problem \eqref{cs_00} with $b = \Phi y_*$ if and only if $\Phi \in \mathbb{R}^{p \times m}$ satisfies NSP of order $L$ for some $\psi \in (0, 1)$. 
NSP of order $L$ is implied by the Restricted Isometry Property (RIP) for a sufficiently small RIP
parameter $\delta_{2L}$~\cite{candes2008restricted}, where a matrix $\Phi \in \mathbb{R}^{p \times m}$ is said to satisfy RIP with constant $\delta_L$ if~\cite{candes2005decoding}  
\begin{equation}\label{cs_rip}
(1 - \delta_L)\norm{y}_2^2 \leq \norm{\Phi y}_2^2 \leq (1+\delta_L)\norm{y}_2^2 \quad	
\mbox{$\forall\, L$-sparse vectors $y \in \mathbb{R}^m$}.
\end{equation}
It is known that RIP is satisfied under many distributional hypothesis on the matrix $\Phi$, for example, random matrices $\Phi$ with entries i.i.d. Gaussian or Bernoulli random variables are known to satisfies RIP with high probability for $L \leq Cm/\log m$ for constant $C$ \cite{baraniuk2008simple, candes2005decoding, candes2006near, rudelson2008sparse}. Recapping, the general pattern of the proof for establishing that sufficiently sparse $y_*$ is the unique minimizer of problem \eqref{cs_00} using distributional assumptions on 
$\Phi$ is given in the following program:
\begin{center}
\text{(CS)}\qquad\begin{tikzcd}
\begin{matrix}\text{Distributional}\\ \text{Assumptions}	
\end{matrix}
  \arrow[r, Rightarrow, "\cite{candes2005decoding}"] & \text{RIP} \arrow[r, Rightarrow, "\cite{candes2008restricted}"] & \text{NSP}\arrow[rr, Leftrightarrow, "\cite{donoho2001uncertainty, gribonval2002sparse}"] && y_* \text{ minimizes \eqref{cs_00}}.
\end{tikzcd}
\end{center}

We extend this program to the class of robust phase retrieval problems
\begin{equation}\label{ic_0}
	\min_x f_p(x) := \norm{|Ax|^p - b}_1
\end{equation}
for $p\in \{1,2\}$, to show that, under Assumption G, and when the residuals  $|Ax_*|^p - b$ are sufficiently sparse,
the vectors $\pm x_*$ are the global minimizers of the real robust phase retrieval problems \eqref{ic_0} with high probability. 
 In our program, we substitute NSP and RIP with new properties called the $p$-Absolute Range Property (p-ARP) and the $p$-Absolute Growth Property (p-AGP), respectively.
%
%

\begin{definition}
[p-Absolute Range Property (p-ARP)] \label{def:ARP}

\noindent
	For $p \in \{1,2\}$, we say $A\in\Rmn$ satisfies the p-Absolute Range Property of order $L_p$ for $\psi_p \in (0,1)$ 
	 if, for any $x, y \in \mathbb{R}^n$ and for any $T \subseteq [m]$ with $|T| \leq L_p$,
	\begin{equation}\label{eq:p-ARP}
		\norm{(|Ax|^p - |Ay|^p)_T}_1 \!\leq\! \psi_p \norm{(|Ax|^p - |Ay|^p)_{T^c}}_1
		\
		\forall\, x, y \in \mathbb{R}^n\,\text{and }
		T \!\subseteq\! [m]\text{ with }|T|\! \leq \!L_p.
	\end{equation}
\end{definition}
In order for Definition~\ref{def:ARP} to make sense, $m$ must be significantly
larger than $n$. This is illustrated by the following example.

\begin{example}
	For  $p\in\{1,2\}$, an example in which ARP does not hold for
	any order $L$ is $A = I_n$ for any $\psi \in (0,1)$. An example in which ARP of order $L=1/3$ holds is $A = (I_n, I_n, I_n)^T$ 
	 for any $\psi \in  [\frac{1}{2},1)$.
\end{example}

	The connection between $p$-ARP and NSP is seen by observing 
	the parallels between \eqref{eq:p-ARP} the fact that
	$\Phi$ satisfies NSP of order $L$ for $\psi \in (0,1)$ \eqref{nsp} if 
\[
	\norm{(Ax - Ay)_T}_1 \leq \psi \norm{(Ax - Ay)_{T^c}}
	\quad \forall\, x, y \in \mathbb{R}^n\text{ and }T \subseteq [m]
	\text{ with }|T| \leq L,
	\]
where the columns of $A$ form a basis of $\text{Null}(\Phi)$.
		
\begin{definition}[p-Absolute Growth Property (p-AGP)]\label{def:agp} 
For $p \in \{1, 2\}$, we say that the matrix $A \in \Rmn$ satisfies the p-Absolute Growth Property if there exists constants $0<\mu_1<\mu_2< 2\mu_1$ and a mapping $\map{\phi_p}{\Rn \times\Rn}{\R_+}$ such that
\begin{equation}\label{eq:agp}
	\mu_1 \phi_p(x, y) \leq \frac{1}{m} \norm{|Ax|^p - |Ay|^p|}_1 \leq \mu_2 \phi_p(x, y)
	\quad \forall\, x, y \in \mathbb{R}^n.
\end{equation}
\end{definition}

The mapping $\phi_p$ is introduced to accommodate the fact that the robust phase retrieval problem cannot have unique solutions since if $x_*$ solves \eqref{ic_0} then so does $-x_*$. For this reason, \eqref{eq:agp} implies that if $x = \pm y$, then $\phi_p(x, y) = 0$. In what follows, we take
\begin{equation}\label{phi_p defined}
	\phi_2(x, y) := \norm{xx^T - yy^T}_F\ \mbox{ and }\ 
\phi_1(x,y) := \min\{\norm{x+y}, \norm{x-y}\}\quad\, \forall\, x,y\in\Rn.
\end{equation}

The relationship between RIP and p-AGP is now seen by comparing \eqref{cs_rip} with \eqref{eq:agp}. 
A fundamental (and essential) difference is that RIP for compressed sensing applies to any selection of $L$ columns from $\Phi$ where $L$ is considered to be small since it determines the sparsity of the solution. On the other hand, our p-AGP
applies to the rows of $A$ corresponding to the zero entries in the sparse residual vector $|Ax_*|^p-b$.



We can now more precisely describe how our program of proof parallels the one used for compressed sensing. 

\begin{enumerate}
\item[$p=2$]:
\begin{tikzcd}
\begin{matrix}\text{\small G}
\end{matrix}  
\arrow[rr, Rightarrow, "\mathrm{Lem} \ \ref{ic_lm_0}"] 
&& 
\text{RIP}\Rightarrow \text{2-AGP} \arrow[rr, Rightarrow, "\mathrm{Lem} \ \ref{ic_lm_main}"] 
&& 
\text{2-ARP}\arrow[rr, Rightarrow, "\mathrm{Thm} \ \ref{g_thm}"] 
&& 
\begin{matrix}x_* \text{\small minimizes}\\ f_2(x)	
\end{matrix}
\end{tikzcd}

\item[$p=1$]:
\begin{tikzcd}
\begin{matrix}\text{\small G} 
\end{matrix}  
\arrow[rr, Rightarrow, "\mathrm{Lem} \ \ref{lm:new6}"] && \text{1-AGP} \arrow[rr, Rightarrow, "\mathrm{Lem} \ \ref{ic_lm_main_2}"] && \text{1-ARP}\arrow[rr, Rightarrow, "\mathrm{Thm} \ \ref{g_thm}"] && 
\begin{matrix}x_* \text{\small minimizes}\\ f_1(x)\end{matrix}
\end{tikzcd}
\end{enumerate}

\section{Global minimization under p-ARP}\label{sec:global} 

In this section we parallel the discussion given 
in \cite{daubechies2010iteratively} with NSP replaced by p-ARP.
 We begin by introducing a measure of residual sparsity. For a vector $y \in \mathbb{R}^n$, let $T\subseteq [m]$ be the set of indices corresponding to the $L$ largest entries in 
 the residual vector $||Ax|^p - b|$ and define 
 $$\sigma^p_L(x) := \norm{(|Ax|^p - b)_{T^c}}_1.$$ 
 Note that $\sigma_L^p(x) = 0$ if and only if $\norm{|Ax|^p - b}_0 \leq L$. 
\begin{lemma}\label{ic_lm_1}
Let $A\in\Rmn$, $p\in\{1,2\}$ and $L\in(0,m)$. If the matrix $A$ satisfies p-ARP of order $L$ for $\psi \in (0, 1)$, then
	\begin{equation}\label{eq:reverse_triangle}
		\norm{|Ax|^p - |Ay|^p}_1 \leq \frac{1+\psi}{1-\psi} 
		\left(\norm{|Ax|^p - b}_1 - \norm{|Ay|^p - b}_1 + 2\sigma^p_{L}(y)\right),
	\end{equation}
for all $x, y \in \mathbb{R}^n$.
\end{lemma}
\begin{proof} In either case 1 or 2 above, let $T$ be the set of indices of the $L$ largest entries in $||Ay|^p - b|$. Then
	\begin{equation}\label{ic_4}\begin{aligned}
		\norm{(|Ax|^p - |Ay|^p)_{T^c}}_1 & \leq \norm{(|Ax|^p - b)_{T^c}}_1 + \norm{(|Ay|^p - b)_{T^c}}_1  \\
		& = \norm{|Ax|^p - b}_1 -\norm{(|Ax|^p - b)_T}_1 + \sigma^p_{L}(y)  \\
		& = \norm{(|Ay|^p - b)_T}_1 - \norm{(|Ax|^p - b)_T}_1 \\
		&\qquad\qquad + \norm{|Ax|^p - b}_1  - \norm{|Ay|^p - b}_1  + 2\sigma^p_{L}(y)  \\
		& \leq \norm{(|Ax|^p - |Ay|^p)_T}_1 + \norm{|Ax|^p - b}_1  - \norm{|Ay|^p - b}_1 + 2\sigma^p_{L}(y) .
		\end{aligned}
	\end{equation}
	By p-ARP,
	\begin{equation} \label{ic_3}
		\norm{(|Ax|^p - |Ay|^p)_T}_1 \leq \psi \norm{(|Ax|^p - |Ay|^p)_{T^c}}_1.
	\end{equation}
 Consequently, by \eqref{ic_4} and \eqref{ic_3},
	\begin{equation}\label{ic_10}
		\norm{(|Ax|^p - |Ay|^p)_{T^c}} \leq \frac{1}{1-\psi}(\norm{|Ax|^p - b}_1 - \norm{|Ay|^p - b}_1 + 2\sigma^p_{L}(y)).
	\end{equation}
	By \eqref{ic_3}, we know 
	\[\begin{aligned}
	\norm{|Ax|^p - |Ay|^p} 
	&= \norm{(|Ax|^p - |Ay|^p)_T}_1 + 
	\norm{(|Ax|^p - |Ay|^p)_{T^c}}_1 \\
	&\leq (1+\psi) 
	\norm{(|Ax|^p - |Ay|^p)_{T^c}}_1.
	\end{aligned}
	\] 
	By combining this with \eqref{ic_10}, we obtain \eqref{eq:reverse_triangle} which holds true for all
	$x,y\in \Rn$.
\end{proof}

The main result of this section now follows.
\begin{theorem}\label{g_thm}
Let $L\in(0,m)$, $p \in \{1,2\}$, and suppose $x_*\in\Rn$ is such that $(|Ax_*|^p - b)$ is $L$ sparse. Let the assumptions of Lemma \ref{ic_lm_1} holds.
	Then $x_*$ is a global minimizer of the robust phase retrieval problem \eqref{ic_0}. Moreover, for any $x$, 
		\[
			\norm{|Ax|^p - |Ax_*|^p}_1 \leq \frac{2(1+\psi)}{1-\psi} \sigma^p_L(x).
		\]
If $\tilde{x}$ is another global minimizer, then $|A\tilde{x}| = |Ax_*|$. If it is further assumed that the entries of $A$ are i.i.d. standard Gaussians and $m \geq 2n-1$, then, with probability 1, $x_*$ is the unique solution of \eqref{ic_0} up to multiplication by $-1$. 
\end{theorem}
\begin{proof}
	By lemma \ref{ic_lm_1}, since $\sigma^p_L(x_*) = 0$,  
	\begin{equation}\label{s}
		\norm{|Ax|^p - |Ax_*|^p}_1 \leq \frac{1+\psi}{1-\psi} (\norm{|Ax|^p - b}_1 - \norm{|Ax_*|^p - b}_1)\quad\forall\, x\in\Rn,
	\end{equation}
	and so $\norm{|Ax|^p - b}_1 \geq \norm{|Ax_*|^p - b}_1$ for all $x$, i.e., $x_*$ is a global minimizer. Again by Lemma \ref{ic_lm_1}, 
	\begin{align}
		\norm{|Ax|^p - |Ax_*|^p}_1 & \leq \frac{1+\psi}{1-\psi}(\norm{|Ax_*|^p - b}_1 - \norm{|Ax|^p - b}_1 + 2\sigma^p_L(x)) \\
		& \leq \frac{2(1+\psi)}{1-\psi} \sigma^p_L(x)
 	\end{align}
Inequality \eqref{s} also implies that if there is another minimizer $\tilde{x}$, then $|Ax_*| = |A\tilde{x}|$. 
	The final statement on the uniqueness of $x_*$ is established in 
\cite[Corollary 2.6]{balan2006signal}. 
	\end{proof}

In the next section we show that under Assumption G, p-ARP of order $L = sm$ holds for a sufficiently small constant $s$, with high probability.
	
\section{Assumption G $\implies$ p-AGP $\implies$ p-ARP}
In this section we use of the Gaussian Assumption G on the matrix $A$ to show that p-AGP holds for $A$ with high probability, and that p-AGP implies p-ARP  of order $L:=sm$ with high probability for a constant $s \in (0, 1)$. The cases $p=2$ and $p=1$ are treated separately since different techniques are required.
\subsection{$\mathbf{p = 2}$}
We begin by re-stating \cite[Lemma 1]{chen2015solving} in our notation, where the conclusion of \cite[Lemma 1]{chen2015solving} is called RIP in \cite{charisopoulos2019low}.

\begin{lemma}[Assumption G $\implies$ 2-AGP(RIP)]
\cite[Lemma 1]{chen2015solving}\label{ic_lm_0}
Under Assumption G, there exists universal constants $c_0, c_1, C$ such that for $\epsilon \in (0, 1)$, if $m > c_0n\eps^{-2}\log \frac{1}{\epsilon}$, then with probability at least $1-C\exp(-c_1 \eps^2 m)$,
\begin{equation} \label{new_1}
	0.9(1-\eps)\norm{M}_F \leq \frac{1}{m}\sum_{i=1}^m|A_iMA_i^T| \leq \sqrt{2}(1+\eps)\norm{M}_F
\end{equation}
for all symmetric rank-2 matrices $M$ which implies 2-AGP with $M = xx^T - yy^T$, $\mu_1=0.9(1-\eps)$ and $\mu_2 = \sqrt{2}(1+\eps)$.
	\end{lemma}
	
\begin{lemma}[Assumption G $\implies$ 2-AGP $\implies$ 2-ARP]\label{ic_lm_main}
Under assumption G, there exist universal constants 
$c_0, c_1,C > 0, s\in (0, 1), \psi \in (0, 1)$
such that if $m > c_0 n$ and $A\in\Rmn$ satisfies G, 
then
\[
	\norm{(|Ax|^2 - |Ay|^2)_T}_1 \!\leq\! \psi \norm{(|Ax|^2 - |Ay|^2)_{T^c}}_1
		\
		\forall\, x, y \in \mathbb{R}^n\,\text{and }
		T \!\subseteq\! [m]\text{ with }|T|\! \leq \!sm
\]
with probability at least $1 - C\exp(-c_1 m)$. Consequently, 2-ARP holds for $m$ with high probability for $m$ sufficiently large.
\end{lemma}
\begin{proof}
We first derive conditions on
	$\eps,\, s \in (0, 1)$ so that $\psi\in(0,1)$ exists.  
	To this end let $\eps,\, s \in (0, 1)$ be given.
	Let $T\subset[m]$ be any subset of $sm$ indices and denote
	by $A_{T^c}$ the $(1-s)m \times n$ sub-matrix of $A$ whose rows correspond to 
	the indices in $T^c$. With this notation, we have 
	$|A_{T^c}x|=|Ax|_{T^c}$.
	Also note that the entries of the 
	matrix $A_{T^c}$ satisfy G.
 By Lemma \ref{ic_lm_0}, there exist universal constants $c_0, c_1, C$ such that if $m > c_0 n \eps^{-2} \log \frac{1}{\epsilon}$, then, for $M = xx^T - yy^T$ and each subset $T \subseteq |m|$ with $|T| = sm$, 
\begin{equation} \label{new_2}
	0.9(1-\eps)\norm{xx^T - yy^T}_F \!\leq\! \frac{1}{(1-s)m}\norm{(|A x|^2 - |A y|^2)_{T^c}}_1 
	\!\leq\! \sqrt{2}(1+\eps)\norm{xx^T - yy^T}_F
\end{equation}
fails to hold with probability no greater than 
$C\exp(-c_1 \eps^2 (1-s)m)$, that is,
2-AGP holds for $A_{T^c}$. 
Since there are $${m \choose (1-s)m} = {m \choose sm} \leq \left(e\frac{m}{sm}\right)^{sm} = \left(\frac{e}{s}\right)^{sm}$$ such $T$'s, the event 
\[
B := \{\eqref{new_2} \text{ holds for every } T \subseteq [m] \text{ with }|T| = sm\} \cap \{\eqref{new_1} \text{ holds}\},
\]
satisfies
\begin{align*}
	\mathbb{P}(B) & \geq 1 - C(e/s)^{sm} \exp(-c_1\eps^2 (1-s)m) - C\exp(-c_1\eps^2 m) \\
	& = 1 - C\exp((1+c_1 \eps^2)sm + sm \log(\frac{1}{s}) -c_1 \eps^2 m) - C\exp(-c_1\eps^2 m).
\end{align*}
Choose $\hat s>0$ so that $(1+c_1 \eps^2)\hat s + \hat s\log(\frac{1}{\hat s}) < \frac{c_1}{2} \eps^2$.
Then, for all $s\in(0,\hat s)$,
$\mathbb{P}(B) \geq 1-2C\exp(-(c_1/2) \eps^2 m)$.
Thus, if event $B$ occurs, we have
\begin{equation}
\begin{aligned}
	\norm{(|Ax|^2 - |Ay|^2)_T}_1 &\! =\! \norm{|Ax|^2 \!-\! |Ay|^2}_1\! -\! \norm{(|Ax|^2\! -\! |Ay|^2)_{T^c}}_1 \\
	& \leq \sqrt{2}(1+\eps)m\norm{xx^T\! -\! yy^T}_F \!- \!0.9(1-\eps)(1-s)m \norm{xx^T\! -\! yy^T}_F\\
	& \leq \frac{\sqrt{2}(1+\eps) \!-\! 0.9(1-\eps)(1-s)}{0.9(1-\eps)(1-s)} \norm{(|Ax|^2 \!-\! |Ay|^2)_{T^c}}_1,
\end{aligned}
\end{equation}
where the first inequality follows from \eqref{new_1} applied to the
first term and \eqref{new_2} applied to the second, and the second inequality follows by \eqref{new_2}.
Consequently, as long as $s \in (0, \hat s)$ is chosen so that $\psi := \frac{\sqrt{2}(1+\eps) - 0.9(1-\eps)(1-s)}{0.9(1-\eps)(1-s)} < 1$, the conclusion follows. 
This can be accomplished by choosing $\epsilon$ so that $\frac{\sqrt{2}(1+\eps)}{1.8(1-\eps)} < 1$ 
(or equivalently, $0<\eps < \frac{1.8 - \sqrt{2}}{1.8 + \sqrt{2}}$) and then choosing 
$s\in(0,\, \min\{\hat s,\, 1 - \frac{\sqrt{2}(1+\eps)}{1.8(1-\eps)}\})$.
\end{proof}

\subsection{$\mathbf{p = 1}$} 
This case requires a series of four technical lemmas in order to establish the main results. We list these lemmas below, and 
their proofs are in the appendix (Section \ref{appendix}). 

	\begin{lemma} \label{ic_lm_4}
	Under assumption G, there exist universal constants $C_0, C_1, C_2$ such that for $\tilde\eps>0$ sufficiently small, if $m > C_0 n \tilde\eps^{-4} \log \tilde\eps^{-1}$, then with probability at least $1 - C_1 \exp(-C_2 \tilde\eps^4 m)$, 
	\begin{equation}\label{lc_13}
	(1 - \tilde\eps) \sqrt{\frac{2}{\pi}}\norm{h} \leq \frac{1}{m}\sum_{i=1}^m |A_i h| \leq (1+\tilde\eps)\sqrt{\frac{2}{\pi}}\norm{h}
	\quad \forall \, h \in \mathbb{R}^n.
	\end{equation}
	\end{lemma}
	
\begin{lemma}\label{lm:new4}
Under assumption G, there exists universal constants $\tilde{c}_0,\tilde{c}_1, \tilde{C}$ such that for $\tilde\epsilon$ sufficiently small, if $m > \tilde{c}_0n \tilde\eps^{-2}\log \frac{1}{\tilde\epsilon}$, then with probability at least $1-\tilde{C}\exp(-\tilde{c}_1 \tilde\eps^2 m)$,
	\begin{equation}\label{eq:nn8}
		\frac{1}{m}\sum_{i=1}^m \left||A_i x|^2 - |A_i y|^2\right|^{\frac{1}{2}} \geq 0.77(1 - \tilde\eps) \norm{xx^T - yy^T}^{\frac{1}{2}}_F
		\quad \forall \, x, y \in \mathbb{R}^n.
	\end{equation}
\end{lemma}

\begin{lemma}\label{lm:new1}
	For $x,\ y \in \mathbb{R}^n$, if $x^Ty \geq 0$ (i.e. $\norm{x - y} \leq \norm{x + y}$), then
	\begin{equation}\label{eq:nn1}
		\norm{x+y} + (\sqrt{2} - 1)\norm{x - y} \geq \norm{x} + \norm{y}
	\end{equation}	
\end{lemma}

\begin{lemma}\label{lm:new5}
	For $x, y \in \mathbb{R}^n$, 
	\begin{equation}\label{eq:nn4}
		\sqrt{2}\norm{xx^T - yy^T}_F \geq \norm{x + y}\norm{x - y}
	\end{equation}
\end{lemma}

We first show that if the matrix $A$  satisfies Assumption G, then it satisfies $1-$AGP with high probability.

\begin{lemma}[Assumption G $\implies$ 1-AGP]\label{lm:new6}
Under assumption G, there exist universal constants $\tilde{C}_0, \tilde{C}_1, \tilde{C}_2 > 0$ such that for $\tilde\eps>0$ sufficiently small, if $m > \tilde{C}_0n\tilde\eps^{-4}\log \frac{1}{\tilde\epsilon}$, then with probability at least $1-\tilde{C}_1\exp(-\tilde{C}_2 \tilde\eps^4 m)$,
	\begin{equation}\label{eq:nn33}
	\mu_1 \phi_1(x, y) \leq \frac{1}{m} \norm{|Ax| - |Ay|}_1 \leq 
	\mu_2\phi_1(x, y)\quad\forall\, x,y\in\Rn,
	\end{equation}
where $\phi_1(x, y)$ is defined in \eqref{phi_p defined}, $\mu_1 =\sqrt{\frac{2}{\pi}}(2-\sqrt{2} - \tilde\eps)$ and $\mu_2 = \sqrt{\frac{2}{\pi}}(1+\tilde\eps)$. Consequently, 1-AGP holds with high probability for $m$ sufficiently large.
\end{lemma}
\begin{proof}
By Lemma \ref{ic_lm_4} and Lemma \ref{lm:new4}, there exist universal constant $c_0, c_1, c_2$ such that for $\epsilon$ sufficiently small, if $m > c_0 n \eps^{-4}\log \frac{1}{\epsilon}$, then with probability at least $1 - c_1 \exp(-c_2 \eps^4 m)$, \eqref{lc_13} and \eqref{eq:nn5} hold. Since we can substitute $y$ by $-y$ if
necessary, without loss of generality, we assume $\norm{x-y} \leq \norm{x+y}$. 

The right hand inequality in \eqref{eq:nn33} easily follows by \eqref{lc_13} and triangle inequality 
\[
\norm{|Ax| - |Ay|}_1 \leq \norm{A(x - y)}_1.
\]

For the left hand inequality of \eqref{eq:nn33}, we consider two cases: (1) $\norm{x - y} \leq \norm{x + y} \leq 10\norm{x - y}$, and
(2) $\norm{x+y} \geq 10\norm{x - y}$.

\begin{enumerate}
	\item[(1)] Assume $\norm{x - y} \leq \norm{x + y} \leq 10\norm{x - y}$.
		By \eqref{lc_13}, we know
	\begin{equation}
		\begin{aligned}
			& \frac{1}{m} \norm{|Ax| - |Ay|}_1  = \frac{1}{m} \sum_{i=1}^m ||A_i x| - |A_i y|| \\
			& = \frac{1}{m} \sum_{i=1}^m |A_i (x+y)| + \frac{1}{m}\sum_{i=1}^m |A_i(x - y)| - \frac{1}{m}\sum_{i=1}^m |A_i x| - \frac{1}{m}\sum_{i=1}^m |A_i y|\\
			& \geq \sqrt{\frac{2}{\pi}} \left((1-\eps)\norm{x+y} + (1-\eps)\norm{x - y} - (1+\eps)\norm{x} - (1+\eps)\norm{y} \right) \\
			& \geq \sqrt{\frac{2}{\pi}}((2-\sqrt{2} - \sqrt{2}\eps)\norm{x - y} - 2\eps\norm{x+y})\\
			& \geq \sqrt\frac{2}{\pi} (2 - \sqrt{2} - (\sqrt{2} + 20) \eps)\norm{x - y},
		\end{aligned}
	\end{equation}
	where the second equality is from $||a| - |b|| = |a + b| + |a-b| - |a| - |b|$ for $a, b \in \mathbb{R}$(since if $ab \geq 0$, then $||a| - |b|| = |a - b|$ and $|a+b| = |a| + |b|$ and if $ab < 0$, then $||a| - |b|| = |a + b|$ and $|a - b| = |a| + |b|$), the first inequality is from Lemma \ref{ic_lm_4}
	(with $=h$ successively set to $x+y,\ x-y,\, x$, and $y$), the second inequality uses Lemma \ref{lm:new1} to replace $\norm{x}+\norm{y}$, and the last inequality follows from our
	assumption that $\norm{x + y} \leq 10\norm{x - y}$.
	\item[(2)] Assume $\norm{x+y} \geq 10\norm{x - y}$. We have
	\begin{equation}
	\begin{aligned}
		\frac{1}{m}\norm{|Ax| - |Ay|}_1 & = \frac{1}{m}\sum_{i=1}^m ||A_i x| - |A_i y|| \\
		&  \geq \left(\frac{1}{m}\sum_{i=1}^m ||A_i x|^2 - |A_i y|^2|^{\frac{1}{2}} \right)^2 
		\!\!\Bigg/ \!\!
		\left( \frac{1}{m} \sum_{i = 1}^m (|A_i x| + |A_i y|)\right) \\
		& \geq \sqrt{\frac{\pi}{2}}\frac{0.77^2(1 - \eps)^2\norm{xx^T - yy^T}_F}{(1+\epsilon)(\norm{x} + \norm{y})} \\
		& \geq \frac{0.77^2\sqrt{\pi}(1 - \eps)^2 \norm{x+y}\norm{x-y}}{2(1+\eps)(\norm{x} + \norm{y})} \\
		& \geq \frac{0.77^2\sqrt{\pi}(1 - \eps)^2 \norm{x+y}\norm{x-y}}{2(1+\eps)(\norm{x+y} + (\sqrt{2}-1)\norm{x-y})} \\
		& \geq \frac{5\cdot 0.77^2\sqrt{\pi}(1 - \eps)^2}{(\sqrt{2} + 9)(1+\eps)}\norm{x-y}
	\end{aligned}
	\end{equation}
	where the first inequality is by Cauchy-Schwartz inequality
	applied to the vectors  with
	$u_i=||A_i x| - |A_i y||^{\frac{1}{2}}$ and $v_i = ||A_i x| + |A_i y||^{\frac{1}{2}}$ , the second inequality is by Lemma \ref{lm:new4} and Lemma \ref{ic_lm_4}, the third inequality is by Lemma \ref{lm:new5}, the fourth inequality is by Lemma \eqref{lm:new1} and the last inequality is by $\norm{x+y} \geq 10\norm{x-y}$. When $0<\eps < 0.01$, one can show by  direct computation that
	\[
		\frac{5\cdot 0.77^2\sqrt{\pi}(1 - \eps)^2}{(\sqrt{2} + 9)(1+\eps)} > 0.02 + \sqrt{\frac{2}{\pi}}(2-\sqrt{2}),
	\]
	and so
	\begin{equation}
		\frac{1}{m}\norm{|Ax| - |Ay|}_1 \geq \sqrt{\frac{2}{\pi}}(2-\sqrt{2})\norm{x - y} 
 	\end{equation}
\end{enumerate}
	Consequently,
	\[
		\frac{1}{m}\norm{|Ax| - |Ay|}_1 \geq \sqrt{\frac{2}{\pi}}\left(2-\sqrt{2} - (20 + \sqrt{2})\eps\right))\norm{x - y}.
	\]
By substituting $\epsilon$ with $\tilde\epsilon/ (\sqrt{2} + 20)$ and adjusting $c_0, c_1,c_2$ we arrived at the desired result.	
\end{proof}

\begin{lemma}[Assumption G $\implies$ 1-AGP $\implies$ 1-ARP]\label{ic_lm_main_2}
Under Assumption G, there exist universal constants $c_0, c_1, C > 0, s\in (0, 1), \psi \in (0, 1)$ such that if $m > c_0 n$, 
then
\[
		\norm{(|Ax| - |Ay|)_T}_1 \!\leq\! \psi \norm{(|Ax| - |Ay|)_{T^c}}_1
		\
		\forall\, x, y \in \mathbb{R}^n\,\text{and }
		T \!\subseteq\! [m]\text{ with }|T|\! \leq \!sm
\]
holds with probability at least $1 - C\exp(-c_1 m)$. Consequently, 1-ARP holds with high probability for $m$ sufficiently large.
\end{lemma}
\begin{proof}
	The proof strategy is similar to Lemma \ref{ic_lm_main}. 
	Let $\phi_1(x, y)$ be as defined in \eqref{phi_p defined}.
	Again, we first derive conditions on
	$\eps,\, s \in (0, 1)$ so that $\psi\in(0,1)$ exists.  
	To this end let $\eps,\, s \in (0, 1)$ be given.
	By Lemma \ref{lm:new6}, there exist universal constants $c_0, c_1, C$ such that if $m > c_0 n \eps^{-4} \log \frac{1}{\epsilon}$, then,  for any $x, y \in \Rn$ and each subset $T \subseteq |m|$ with $|T| = sm$, the double sided inequality\begin{equation} \label{ic_7}
		\sqrt{\frac{2}{\pi}}(2-\sqrt{2} - \eps) \phi_1(x, y) \leq \frac{1}{(1-s)m} \norm{(|Ax| - |Ay|)_{T^c}}_1 \leq \sqrt{\frac{2}{\pi}}(1+\eps)\phi_1(x, y)
\end{equation}
fails to hold with probability no larger than $C\exp(-c_1 \eps^2 (1-s)m)$, that is, 1-AGP holds for $A_{T^c}$. We know for the event $B := \{\eqref{ic_7} \text{ holds for every } \text{ with }|T| = sm\} \cap \{\eqref{eq:nn33} \text{ holds}\}$,
 by taking $s$ sufficient small, there exist positive constant constant $\tilde{c}$ and $\tilde{C}$ such that $\mathbb{P}(B) \geq 1-\tilde{C}\exp(-\tilde{c} \eps^4 m)$. On the event $B$, we obtain
\begin{align}
	\norm{(|Ax| - |Ay|)_T}_1 & = \norm{|Ax| - |Ay|}_1 - \norm{(|Ax| - |Ay|)_{T^c}}_1 \nonumber \\
	& \leq \sqrt{\frac{2}{\pi}}(1+\eps)m\phi_1(x, y) \nonumber  - \sqrt{\frac{2}{\pi}}(2-\sqrt{2})(1-\eps)(1-s)m \phi_1(x, y) \nonumber \\
	& \leq \frac{(1+\eps) - (2-\sqrt{2})(1-\eps)(1-s)}{(2-\sqrt{2})(1-\eps)(1-s)} \norm{(|Ax| - |Ay|)_{T^c}}_1 \nonumber 
	\end{align}
So as long as we choose $s \in (0, 1)$ such that $\psi := \frac{(1+\eps) - (2-\sqrt{2})(1-\eps)(1-s)}{(2-\sqrt{2})(1-\eps)(1-s)} < 1$, the conclusion follows. More precisely, $0< s < 1 - \frac{1+\eps}{2(2-\sqrt{2})(1-\eps)}$ (Note $\epsilon$ must be chosen such that $\frac{1+\eps}{2(2-\sqrt{2})(1-\eps)} < 1$ in advance, which is possible since $2(2-\sqrt{2}) > 1$).

\end{proof}


By combining the results of this section with those of Section \ref{sec:global} we show under Assumption G that the solutions to the $\ell_0$ optimization problem \eqref{ctri_l00} and $\ell_1$ optimization problem \eqref{ic_0} coincide with high probability when the residuals are sufficiently sparse. Methods for solving \eqref{ic_0} often require that the objective function $f_p$ satisfies a sharpness condition. In the next section, we consider this sharpness condition.

\section{Sharpness} 
In this section we show that, under assumption G, if $|Ax_*|^p-b$ is
sufficiently sparse, then the function
\[
f_p(x):=\frac{1}{m}\norm{\,|Ax|^p-b}_1
\]
is sharp with respect to the solution set $\{x_*,-x_*\}$ with high probability, for $p=1,2$.
%
%
Sharpness is an extremely useful tool for analyzing the convergence and the rate of convergence of optimization algorithms
\cite{burke2002weak, burke2005weak, burke2009weak, burke1993weak, charisopoulos2019low, davis2017nonsmooth, duchi2017solving}. 

\begin{definition}\cite{burke1993weak}
Let $f: \mathbb{R}^n \rightarrow \mathbb{R}$ and set
$\mathcal{X}:=\argmin f$. Then $f$ is said to be sharp with respect to $\mathcal{X}$ if
 \[
 f(x)\ge \min_x f \ +\ \mu \text{dist}(x. \mathcal{X})\quad \forall 
 x\in\Rn,
 \]
 where $\text{dist}(x. \mathcal{X}):=\inf_{y \in \mathcal{X}} \norm{x - y}$.
\end{definition}

\begin{theorem}\label{g_thm1}
	Let Assumption G hold and let $p\in\{1,2\}$.
	Then there exist constants $C_p, c_{p0}, c_{p1}> 0$ and 
	$s_p\in (0,1)$, such that if $\norm{|Ax_*|^p - b}_0 \leq s_p m$, 
	then, for 
	$m \geq c_{p0} n$, $f_p$ is sharp with probability at least $1-C_p\exp(-c_{p1}m)$.
\end{theorem}
\begin{proof}
Let $C_p, c_{p0}, c_{p1}> 0$ and 
	$s_p\in (0,1)$ be as in Lemma \ref{ic_lm_main} for $p=2$ and
	as in Lemma \ref{ic_lm_main_2} for $p = 1$.
	By either Lemma \ref{ic_lm_main} ($p=2$) or
	Lemma \ref{ic_lm_main_2} ($p=1$), $A$ satisfies $p$-ARP of order $s_p m$ for $\psi_p \in (0, 1)$ for $p = 1, 2$, where $s_p$ and $\psi_p$ are constants depending on $p$. 
	Hence, by \eqref{s},
	\begin{equation}\label{eq:sharp1}
		f_p(x) - f_p(x_*) \geq \frac{1-\psi_p}{m(1+\psi_p)} \norm{|Ax|^2 - |Ax_*|^2}_1
	\end{equation}
	For $p = 2$, Lemma \ref{ic_lm_0} tells us that if $m \geq c_{p0} \eps^{-2}\log (\frac{1}{\epsilon})n$, then, with probability at least $1- C_p\exp\left( c_{p1} \eps^{-2}\log \frac{1}{\epsilon}m\right)$, 
	\begin{align}
		\frac{1}{m} \norm{|Ax|^2 - |Ax_*|^2}_1 & \geq 0.9 (1-\eps) \norm{xx^T - yy^T}_F \nonumber \\
		& \geq 0.45\sqrt{2}(1-\eps)\norm{x+x_*}\norm{x-x_*} \nonumber \\
		& = 0.45\sqrt{2}(1-\eps)
		\phi_1(x,x_*) 
		\max\{\norm{x-x_*}, \norm{x+x_*}\} \nonumber \\
		& \geq 0.45\sqrt{2}(1-\eps)\norm{x_*} \text{dist}(x, \{x_*, -x_*\}), 
	\end{align}
	where $\phi_1(x,x_*)$ is defined in \eqref{phi_p defined}.
	For $p = 1$, Lemma \ref{lm:new6} tells us that, if $m \geq c_{p0} \eps^{-4}\log (\frac{1}{\epsilon})n$, then, with probability at least $1- C_p\exp\left( c_{p1} \eps^{-4}\log \frac{1}{\epsilon}m\right)$,
	\[
		\frac{1}{m}\norm{|Ax| - |Ax_*|}_1 \geq \sqrt{\frac{2}{\pi}}(2-\sqrt{2} - \eps) \text{dist}(x, \{x_*, -x_*\}) .
	\]
	Thus, in either case, by taking an $0 < \epsilon < 1$ small enough and using \eqref{eq:sharp1}, there is constant $\mu>0$ such that
	\[
		f_p(x) - f_p(x_*) \geq \mu \text{dist}(x, \mathcal{X}), 
	\]
	 where $\mathcal{X}$ is $\argmin f_p$.  
\end{proof}

It is shown in \cite{charisopoulos2019low, duchi2017solving} that if $f_2$ is sharp and weakly convex at $\argmin f_2$, then prox-linear method and subgradient descent method with geometrically decreasing stepsize  converges locally quadratically and locally linearly, respectively. Since weak convexity of $f_2$ under assumption G is already shown in $\cite{duchi2017solving, charisopoulos2019low, davis2017nonsmooth}$, sharpness in this regime guarantees these two algorithms converge with the specified rate. In both algorithms proper initialization is needed (e.g., Section 5 of \cite{zhang2016provable}).

\section{Concluding Remarks}
There are a number of recent results discussing the nature of the solution set to the robust phase retrieval problem $\min_x f_2(x)$ with sparse noise under weaker distributional hypothesis than employed here \cite{charisopoulos2019low, duchi2017solving, zhang2016provable, chen2017robust}. The focus of these works are algorithmic. Their goal is to show their methods are robust to outliers, and, in addition, some establish the sharpness of $f_2$  in order to prove rates of convergence \cite{charisopoulos2019low, duchi2017solving}. 
Although these works use weaker distributional hypothesis, the probability of successful recovery is an average over all possible subsets $T \subseteq [m]$ with $|T| = sm$ for some $s \in (0, \frac{1}{2})$. Consequently, the value of $s$ in their results is larger than ours. The reason for this difference is that, in our result, successful recovery is valid for all possible subsets $T \subseteq [m]$ with $|T| = sm$ for some $s \in (0, 1)$, with uniformly high probability. A more precise description is this difference follows. 

In \cite{charisopoulos2019low, duchi2017solving}, the random matrix $A$ and the random index set $T \subseteq [m]$, with $|T| = sm$ for $s \in (0, \frac{1}{2})$, are drawn independently of each other. Let $w \in \{0, 1\}^m$ denote the random indicator vector of $T$, that is, $w_i = 1$ if $i \in T$ and $w_i = 0$ otherwise. Let $z \in \mathbb{R}^m$ be an arbitrary vector. The noisy model in \cite{charisopoulos2019low, duchi2017solving} has the form
\[
	\min_x \tilde{f}_2(x) :=  \norm{|Ax|^2 - (\vec{1} - w) \odot b - w \odot z}_1,
\]
where $b = |Ax_*|^2$, $\vec{1}$ represents the vector with $1$ in each entry and $\odot$ represents the elementwise product of vectors. The authors in \cite{charisopoulos2019low, duchi2017solving} prove sharpness of $\tilde{f}_2$ with respect to $x_*$ with high probability. Due to the independence of $A$ and $T$, in fact, they show that the probabillity
\[
	\mathbb{P}(\tilde{f}_2 \text{ is sharp}) = \frac{1}{{m\choose sm}}\sum_{T_0:|T_0| = sm} \mathbb{P}(\tilde{f}^{T_0}_2 \text{ is sharp})
\]
is high, where $\tilde{f}^{T_0}_2(x) := \norm{|Ax|^2 - (\vec{1} - w_0) \odot b - w_0 \odot z}_1$ and $w_0$ is the indicator vector for a \emph{fixed} index set $T_0$. On the other hand, we show that with high probability, $\tilde{f}^{T_0}_2$ is sharp for \emph{all} possible $T_0$ with $|T_0| = sm$. Our result is a stronger implication, however, it comes at the expense of a smaller value for $s$. By design, this result closely parallels the result in \cite{candes2005decoding} for compressed sensing.

\section{Appendix}\label{appendix}
In this appendix we provide the proofs for Lemmas \ref{ic_lm_4}, \ref{lm:new4}, \ref{lm:new1}, 
and \ref{lm:new5}.
These proofs make use of a 
Hoeffding-type inequality \cite{vershynin2010introduction} explained below.
A random variable $X$ is said to be sub-gaussian
\cite[Definition 5.7]{vershynin2010introduction} if 
\begin{equation}\label{eq:define psi}
\norm{X}_{\psi_2}:=\sup_{p \geq 1} p^{-1/2}(\mathbb{E}|X|^p)^{1/p}
\end{equation}
 is finite, and is said to be
centered if it has zero expectation.
By \cite[Proposition 5.10]{vershynin2010introduction},
there is a universal constant $c>0$ such that
if $X_1,..., X_N$ are independent centered sub-gaussian random variables, 
then, for every 
$a = \{a_1, ..., a_N\} \in \mathbb{R}^N$ and $t \geq 0$, we have
	\begin{equation}\label{ic_lm_3}
	\mathbb{P}\left(|\sum_{i=1}^N a_i X_i| \geq t\right) \leq e \cdot \exp(-\frac{ct^2}{K^2\norm{a}^2}),
	\end{equation}
	where $K := \max_i \norm{X_i}_{\psi_2}$.
	\medskip

\noindent
{\bf Proof of Lemma \ref{ic_lm_4}}:
First observe that the inequality \eqref{lc_13} is trivially true for $h=0$.
Next, let $h\in\Rn\setminus\{0\}$ and $0 < \eps < \sqrt{2} - 1$.  Observe that $\frac{|A_i h|}{\norm{h}}$ are independent sub-gaussian random variables with mean $\sqrt{\frac{2}{\pi}}$. 
	Therefore, $\frac{|A_ih|}{\norm{h}} - \sqrt{\frac{2}{\pi}}$ is a centered sub-gaussian random variable. Hence, \eqref{ic_lm_3} tells us that 
	there are universal constants $C>0$ and $c_0>0$ such that
	\begin{equation}\label{eq:each h}
	\mathbb{P}\left(\left|\sum_{i=1}^m \left(\frac{|A_i h|}{\norm{h}} - \sqrt{\frac{2}{\pi}}\right)\right| > m\sqrt{\frac{2}{\pi}} \eps\right) \leq C\exp(-c_0m\eps^2).
	\end{equation}
Therefore \eqref{lc_13} holds for each fixed $h\in \Rn\setminus\{0\}$
with probability $1 - C\exp(-c_0m\eps^2)$. 
We now show that there exist a universal event with large probability, in which \eqref{lc_13} holds for every $h$. 
On the unit sphere $S:=\{x| \norm{x} = 1\}$ construct an $\eps$-net $\mathcal{N}_{\eps}$ with $|\mathcal{N}_{\epsilon}| \leq (1+\frac{2}{\epsilon})^n$
\cite[Lemma 5.2]{vershynin2010introduction}, i.e., for any $h \in S$, 
there exists $h_0 \in \mathcal{N}_{\eps} \subseteq S$ 
such that $\norm{h - h_0} \leq \epsilon$. 
Taking the probability of the union of the events in 
\eqref{eq:each h}
for all the points $h_0\in \mathcal{N}_{\eps}$, we 
obtain the bound $C(1+\frac{2}{\epsilon})^n \exp(-c_0 m\eps^2)$.
Hence,
 \eqref{lc_13} holds for each $h_0 \in \mathcal{N}_{\eps}$ with probability at least $1-C(1+\frac{2}{\epsilon})^n \exp(-c_0 m\eps^2)$. On the intersection of these events and the event of Lemma \ref{ic_lm_0}, we deduce, for any $h$ with $\norm{h} = 1$,
 \begin{equation}
		\begin{aligned}
		\frac{1}{m}|\sum_{i=1}^m |A_i h| - \sum_{i=1}^m |A_i h_0|| & \leq \frac{1}{m}\sum_{i=1}^m ||A_i h| - |A_i h_0||  \\
		& \leq 	\frac{1}{m}\sum_{i=1}^m ||A_i h|^2 - |A_i h_0|^2|^{\frac{1}{2}} \\
		& \leq (\frac{1}{m} \sum_{i=1}^m ||A_i h|^2 - |A_i h_0|^2|)^{\frac{1}{2}}  \\
		& \leq 2^{1/4}(1+\eps)^{1/2} \norm{hh^T - h_0h_0^T}_F^{\frac{1}{2}}  \\
		& \leq 2^{1/4}(1+\eps)^{1/2}(\norm{h - h_0}\norm{h} + \norm{h - h_0}\norm{h_0})^{\frac{1}{2}} \\
		& \leq 2^{5/4}\epsilon^{1/2}, 
		\end{aligned}
\end{equation}
		where the second inequality follows since $||a| - |b||^2 \leq(|a|+|b|)||a|-|b||)$, the third from the concavity of $(\cdot)^2$, the fourth is by Lemma \ref{ic_lm_0}, the fifth is by triangle inequality and the last inequality is from $\norm{h} = \norm{h_0} = 1$ and $\norm{h - h_0} \leq \eps$.
		Hence
		\[
		(1 - \eps - 2^{3/4}\sqrt{\pi \epsilon})\sqrt{\frac{2}{\pi}} \leq \frac{1}{m} \sum_{i=1}^m |A_i h| \leq (1 + \eps + 2^{3/4}\sqrt{\pi \epsilon})\sqrt{\frac{2}{\pi}}
		\]
		holds for all $\norm{h} = 1$ with probability at least $1 - (1+\frac{2}{\epsilon})^n\exp(-c_0 m \eps^2) - c_2\exp(-c_3 m \eps^2)$, for $m \geq c_1 n \epsilon^{-2} \log(\frac{1}{\epsilon})$. For $c_1 > 0$ sufficiently large and $\epsilon$ small, the probability is at least 
		\begin{equation}\label{eq:change epsilon}
		\begin{aligned}
		& 1 - c_2 \exp(-c_3 m \eps^2) - \exp(-c_0 m \eps^2 + 2n\log(\frac{1}{\epsilon})) \\
		& \geq 1-c_2 \exp(-c_3 m \eps^2) - \exp(-(c_0 - \frac{2}{c_1})m\eps^2) \\
		& \geq 1-\tilde{c}_2\exp(-\tilde{c}_3 m \eps^2),		
		\end{aligned}
		\end{equation}
for some $\tilde{c}_2,\, \tilde{c}_3>0$. 
		 By letting $\tilde{\epsilon} = \eps + 2^{3/4}\sqrt{\pi \epsilon}<(1+2^{3/4}\sqrt{\pi})\sqrt{\eps}$ 
		 so that $\eps \ge k\tilde \eps^2$ for $k>0$, we arrive at the desired result. 
\hfill $\square$
\medskip

\noindent
{\bf Proof of Lemma \ref{lm:new4}}:
We only need to prove 
\begin{equation}\label{eq:nn5}
	\frac{1}{m}\sum_{i=1}^m \left|A_iMA_i^T\right|^{\frac{1}{2}} \geq 0.77(1-\eps)\norm{M}^{\frac{1}{2}}_F
\end{equation}
holds for all rank-2 matrix $M$ with high probability. Clearly this inequality holds when $M = 0$. Assume $M \neq 0$. Furthermore, since we can divide \eqref{eq:nn5} by $\norm{M}^{\frac{1}{2}}$ on both sides, we can assume $\norm{M} = 1$. Moreover, using the eigenvalue decomposition of $M$, we can assume that $M = z_1z_1^T - sz_2z_2^T$ where $z_1^Tz_2 = 0$, $\norm{z_1} = \norm{z_2} = 1$ and $s \in [-1, 1]$. Since for each $i$, 
$A_i z_1$ and $A_i z_2$ are independent standard gaussians, 
\begin{equation}\label{eq:ama}
\left|A_i MA_i^T\right|^{\frac{1}{2}} = \left|(A_i z_1)^2 - s(A_i z_2)^2\right|^{\frac{1}{2}} \leq \left((A_i z_1)^2 + (A_i z_2)^2\right)^{\frac{1}{2}} \leq |A_iz_1| + |A_i z_2|
\end{equation} 

are sub-gaussian. Set $e(s) := \mathbb{E}\left|A_i M A_i^T\right|^{\frac{1}{2}} = \mathbb{E}\left|Z_1^2 - s Z_2^2\right|^{\frac{1}{2}}$ where $Z_1$ and $Z_2$ are independent standard gaussian scaler random variables.
Notice $\norm{M}_F = \norm{z_1z_1^T - s z_2z_2^T}_F = \sqrt{1+s^2}$ and 
	\begin{equation}
		\begin{aligned}
			e(s) = \mathbb{E}\left|Z_1^2 - sZ_2^2\right|^{\frac{1}{2}} & = \frac{1}{2\pi}\int_0^\infty r^2e^{-\frac{r^2}{2}}dr\int_0^{2\pi}\left|\cos^2\theta - s \sin^2 \theta\right|^{-\frac{1}{2}}d\theta \\
			& = \frac{1}{2\sqrt{2\pi}}\int_0^{2\pi}\left|\cos^2\theta - s \sin^2 \theta\right|^{\frac{1}{2}}d\theta
		\end{aligned}
	\end{equation}
	We draw a plot of $\frac{e(s)}{\norm{M}_F} = \int_0^{2\pi} \left|\cos^2\theta - s \sin^2 \theta\right|^{\frac{1}{2}}d\theta /(2\sqrt{2\pi(1+s^2)})$ when $s \in [-1, 1]$ through a numerical experiment.
	\begin{figure}[H]
	\centering
	\includegraphics[width = 0.8\textwidth]{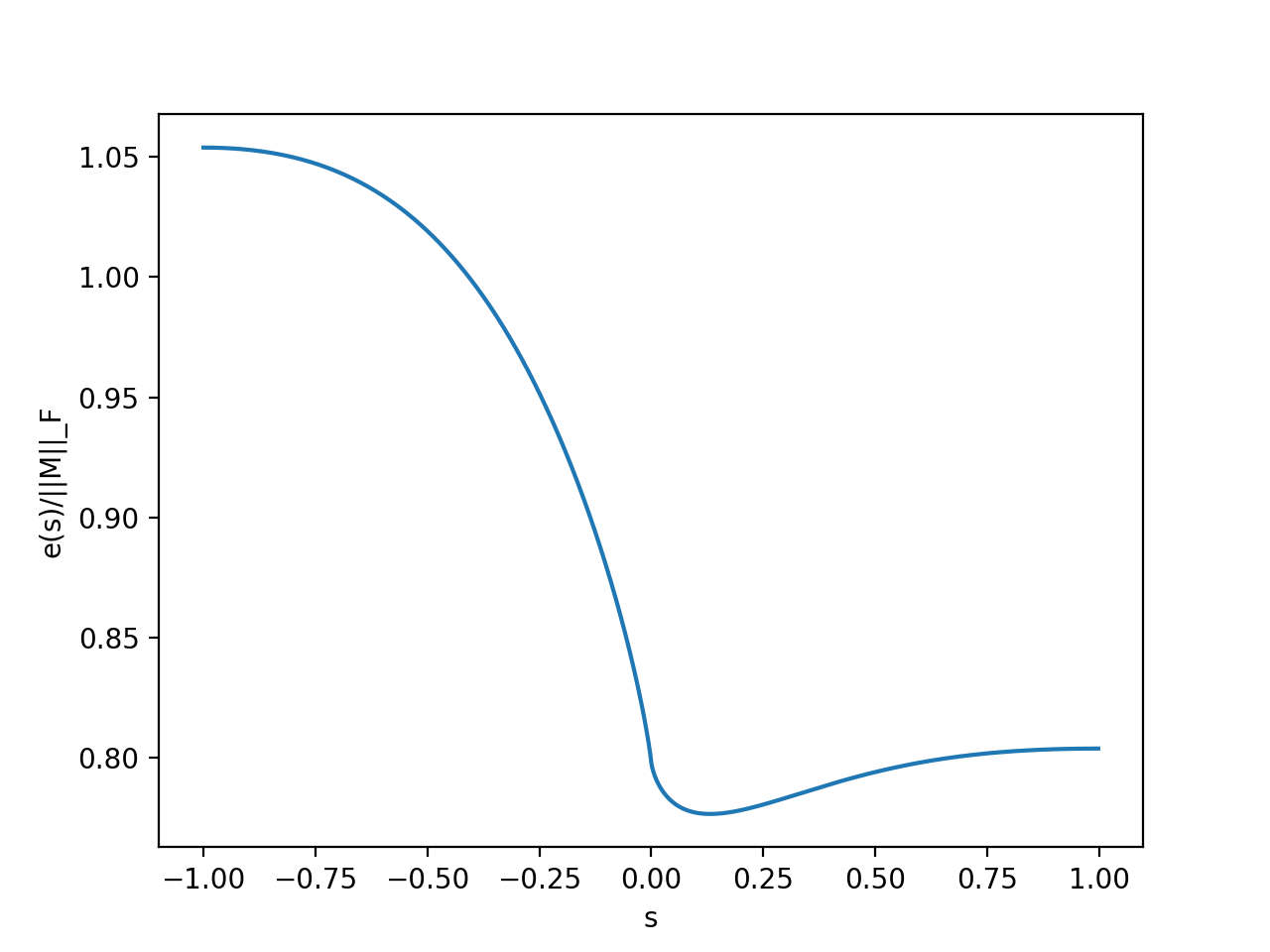}
	\caption{Values of $\frac{e(s)}{\norm{M}}_F$ when $s \in [-1, 1]$.} 
\end{figure}
Numerical experiment above shows that $\frac{e(s)}{\norm{M}_F} \geq 0.77$ (hence $e(s) \geq 0.77$) for all $a \in [-1, 1]$. Note that for each $i$, $ Y_i := \frac{\left|A_i M A_i^T \right|^{\frac{1}{2}}}{e(s)} - 1$ is a centered sub-gaussian random variable. Hence, by \eqref{eq:define psi} and \eqref{eq:ama},
\[
	\norm{Y_i}_{\psi_2} \leq \sup_{p \geq 1}p^{-\frac{1}{2}}(\frac{2(\mathbb{E}|Z|^p)^{\frac{1}{p}}}{e(s)} + 1) \leq \frac{2}{0.77}\norm{Z}_{\psi_2} + 1 < +\infty.
\] 
where $Z$ is a standard gaussian variable.
Hence, \eqref{ic_lm_3} tells us that there exist universal constants $C>0$ and $c_0>0$ such that
	\begin{equation}\label{eq:each s1}
	\mathbb{P}\left(\left|\sum_{i=1}^m \left(\frac{\left|A_i M A_i^T\right|^{\frac{1}{2}}}{e(s)} - 1\right)\right| > m \eps\right) \leq \hat{C}\exp(-\hat{c}_0m\eps^2)
	\end{equation}
Consequently, for fixed $M$,
	\begin{equation}\label{eq:each s2}
		\frac{1}{m}\sum_{i=1}^m \left|A_i M A_i^T \right|^{\frac{1}{2}} \geq (1-\eps)e(s) \geq 0.77(1-\eps)\norm{M}_F
	\end{equation}
	holds with probability at least $1-\hat{C}\exp(-\hat{c}_0 m \eps^2)$. 
	
Next we generalize \eqref{eq:each s2} to all rank-2 matrices $M$. Again, by scale invariance, we assume $\norm{M}_F = 1$. Consequently, we only need to prove \eqref{lm:new4} holds with high probability for all $M \in \mathcal{M} := \{\beta uu^T + \gamma vv^T| \norm{u} = \norm{v} = 1$, $u^T v = 0$ \text{ and } $\beta^2 + \gamma^2 = 1\}$. Set $\mathcal{S}_{\epsilon^2} := \mathcal{T}_{\epsilon^2} \times \mathcal{N}_{\epsilon^2} \times \mathcal{N}_{\epsilon^2}$ where $\mathcal{T}_{\epsilon^2}$ is an $\epsilon^2$-net of $[-1, 1]$ and $\mathcal{N}_{\eps^2}$ is an $\epsilon^2$-net of the unit sphere $\{x \in \mathbb{R}^n|\norm{x} = 1\}$. Since $|\mathcal{T}_{\eps^2}| \leq \frac{2}{\epsilon^2}$ and $|\mathcal{N}_{\epsilon^2}| \leq \left(\frac{3}{\epsilon^2} \right)^n$, we know $|\mathcal{S}_{\epsilon^2}| \leq \left(\frac{3}{\epsilon} \right)^{4n+2}$. Let $E$ denote the event that  \eqref{eq:each s2} holds for every $(\beta_0, u_0. v_0) \in \mathcal{S}_{\epsilon^2}$. Consequently,

	\centerline{
		$\mathbb{P}(E) \geq 1 - 2\hat{C}\left(\frac{3}{\epsilon}\right)^{4n+2}\exp(-\hat{c}_0 m \eps^2).$
	}
	
	\noindent
	For $M \in \mathcal{M}$, we want to approximate $M = \beta uu^T + \gamma vv^T$ by an element $M_0 = \beta_0 u_0u_0^T + \gamma_0 v_0v_0^T \in \mathcal{M}$ with $(\beta_0, u_0, v_0) \in \mathcal{S}_{\epsilon^2}$. More precisely, let $(\beta_0, u_0, v_0) \in \mathcal{S}_{\epsilon^2}$ and $M_0 = \beta_0 u_0 u_0^T + \sgn(\gamma)\sqrt{1-\beta_0^2} v_0 v_0^T$ be such that $|\beta - \beta_0| \leq \epsilon^2$, $\norm{u - u_0} \leq \epsilon^2$ and $\norm{v - v_0} \leq \epsilon^2$. Consequently, we have 
	
	\centerline{
		$|\gamma - \sgn(\gamma) \sqrt{1-\beta_0^2}| = |\sqrt{1-\beta^2} - \sqrt{1-\beta_0^2}| \leq \left|\beta^2 - \beta_0^2\right|^{\frac{1}{2}} \leq \sqrt{2}\left|\beta-\beta_0\right|^{\frac{1}{2}} \leq \sqrt{2}\epsilon.
	$}
	Also note that
	\begin{equation}\label{eq:nn6}
	\begin{aligned}
		\norm{\beta uu^T - \beta_0 u_0 u_0^T}_F & \leq |\beta - \beta_0| \norm{uu^T}_F + \norm{\beta_0u(u - u_0)^T}_F + \norm{\beta_0 (u - u_0)u_0^T}_F \\
		& = |\beta - \beta_0|\norm{u}^2 + |\beta_0|\norm{u - u_0}(\norm{u}+\norm{u_0}) \\
		& \leq 3\epsilon^2 < 4\epsilon
	\end{aligned}
	\end{equation}
	Similarly we can prove $\norm{\gamma vv^T - \sgn(\gamma)\sqrt{1-\beta_0^2} v_0v_0^T} \leq 2\eps^2 + 2\epsilon < 4\epsilon$. On the intersection of events where \eqref{new_1} holds and $E$, we have
	\begin{equation*}
		\begin{aligned}
			& \left|\frac{1}{m} \sum_{i=1}^m |A_i M A_i^T|^{\frac{1}{2}}  - \frac{1}{m} \sum_{i=1}^m |A_i M_0 A_i^T|^{\frac{1}{2}} \right|
			  \leq \frac{1}{m}\sum_{i=1}^m \left| \left|A_i M A_i^T\right|^{\frac{1}{2}} - \left|A_i M_0 A_i^T \right|^{\frac{1}{2}} \right| \\
			& \leq \frac{1}{m} \sum_{i=1}^m \left| \left|A_i M A_i^T\right| - \left|A_i M_0 A_i^T\right| \right|^{\frac{1}{2}} \\
			& \leq \left(\frac{1}{m} \sum_{i=1}^m \left| \left|A_i M A_i^T\right| - \left|A_i M_0 A_i^T\right| \right| \right)^\frac{1}{2} \\
			& \leq \left(\frac{1}{m} \sum_{i=1}^m  \left|A_i (M- M_0) A_i^T\right|  \right)^\frac{1}{2} \\
			& \leq \left(\frac{1}{m} \sum_{i=1}^m  \left|A_i (\beta uu^T - \beta_0 u_0u_0^T) A_i^T\right| +  \left|A_i (\gamma vv^T - \gamma_0 v_0v_0^T) A_i^T\right|\right)^\frac{1}{2} \\
			& \leq \left(\frac{1}{m} \sum_{i=1}^m  \left|A_i (\beta uu^T - \beta_0 u_0u_0^T) A_i^T\right| \right)^\frac{1}{2} + \left(\frac{1}{m}\sum_{i=1}^m \left|A_i (\gamma vv^T - \gamma_0 v_0v_0^T) A_i^T\right|\right)^{\frac{1}{2}} \\
			& \leq 2^{\frac{1}{4}}(1+\eps)^{\frac{1}{2}} \norm{\beta u u^T - \beta_0 u_0 u_0^T}_F^{\frac{1}{2}} + 2^{\frac{1}{4}}(1+\eps)^{\frac{1}{2}} \norm{\gamma v v^T - \gamma_0 v_0 v_0^T}_F^{\frac{1}{2}} \\
			& \leq  2^{\frac{9}{4}} (1+\eps)^{\frac{1}{2}}\epsilon^{\frac{1}{2}},
		\end{aligned}
	\end{equation*}
where the second inequality is by $||a| - |b||^2 \leq |a^2 - b^2|$ for any $a, b \in \mathbb{R}$, the third inequality is by concavity of $(\cdot)^2$, the fourth and the fifth inequalities are by triangle inequality, the sixth inequality is by $a^2 + b^2 \leq (a + b)^2$ for any $a, b \in \mathbb{R}$ and the seventh inequality is by the right hand side of equation \eqref{new_1}. Consequently, if $m > c_0 n \eps^{-2}\log \frac{1}{\epsilon}$
 \begin{equation}\label{eq:nn7}
	\frac{1}{m}\sum_{i=1}^m \left|A_iMA_i^T\right|^{\frac{1}{2}} \geq 0.77(1 - \eps - 2^{\frac{9}{4}}(1+\eps)^{\frac{1}{2}} \eps^{\frac{1}{2}})
\end{equation}
holds with probability at least $1 - 2\hat{C}\left(\frac{3}{\eps}\right)^{4n+2}\exp(-\hat{c}_0 m \eps^2) - C\exp(-c_1 \eps^2 m)$. As in \eqref{eq:change epsilon}, by making $c_0$ large, we are able to make the probability $\geq 1 - \hat{\hat{C}}\exp(-\hat{\hat{c}}_0 m \eps^2)$ for some constants $\hat{\hat{C}}$ and $\hat{\hat{c}}_0$. 
By letting $\tilde\epsilon := \eps + 2^{\frac{9}{4}}(1+\eps)^{\frac{1}{2}}\eps^{\frac{1}{2}}$ and adjust constants $\hat{\hat{C}}, \hat{\hat{c}}_0, c_0$ we arrive at the desired result. 
\hfill$\square$

\noindent
{\bf Proof of Lemma \ref{lm:new1}}:
If $x=0$ or $y=0$ or $x=y$, the inequality holds. 
Thus, in particular, by
the symmetry of \eqref{eq:nn1} in $x$ and $y$, we can assume that 
$\norm{x}\ge\norm{y}>0$. Dividing \eqref{eq:nn1} by $\norm{x}$, tells us that we can assume
$\norm{x} = 1$ and $\norm{y} = t$ for $t \in [0, 1]$. 
Set $\rho:=\frac{x^Ty}{\norm{y}} \in [0, 1]$, and define 
$h(t,\rho) := \sqrt{t^2 -2\rho t + 1} + \sqrt{t^2 + 2\rho t + 1} -1 -t = \norm{x+y} + \norm{x-y} - \norm{x} - \norm{y}$. 
If $x = y$, we are done; otherwise, set 
$q(t, \rho) := \frac{h(t,\rho)}{\sqrt{t^2 -2\rho t + 1}} = \frac{\norm{x+y} + \norm{x-y} - \norm{x} - \norm{y}}{\norm{x - y}}$, for each $(t,\rho)\in[0,1]\times[0,1]$. We now show that the minimum value of $q$ over $[0,1]\times[0,1]$ is $2-\sqrt{2}$.
	For fixed $t\in[0,1]$,
	\begin{equation*}
	\begin{aligned}
		\frac{\partial q(t, \rho)}{\partial \rho} & = \frac{t(t^2+1)}{(t^2 - 2\rho t + 1)^{\frac{3}{2}}} \left[\frac{2}{(t^2 + 2\rho t + 1)^{\frac{1}{2}}} - \frac{t+1}{t^2+1}\right] \\
		& \geq \frac{t(t^2+1)}{(t^2 - 2\rho t + 1)^{\frac{3}{2}}} \left[\frac{2}{t+1} - \frac{t+1}{t^2+1}\right]\\
		& \geq 0,
	\end{aligned}
	\end{equation*}
	where the first inequality follows since $t^2 + 2\rho t + 1 \leq (1+t)^2$ as $\rho\in[0,1]$, and the last inequality follows since $2(t^2+1) \geq (t+1)^2$. That is, $q(t, \rho)$ is increasing with respect to $\rho$ when $\rho \in [0, 1]$ for each fixed $t \in [0, 1]$. Also
	\[
		\frac{d q(t, 0)}{dt} = -\frac{1-t}{(1+t^2)^{\frac{3}{2}}} \leq 0.
	\]
	Hence $q(t, 0)$ is decreasing for $t \in [0, 1]$. We know for each $t \in [0, 1]$, $\rho \in [0, 1]$, 
	\[
		q(t, \rho) \geq q(t, 0) \geq q(1, 0) = 2-\sqrt{2} 
	\]
	Thus $h(t, \rho) \geq (2-\sqrt{2})\norm{x-y}$, which leads to the desired result.
\hfill$\square$

\noindent
{\bf Proof of Lemma \ref{lm:new5}}:
	If $x = y = 0$, we are done. Next assume at least one of $x$ and $y$ is non-zero. We assume $\norm{x} = 1$ and $\norm{y} = t \in [0, 1]$ since we can divide \eqref{eq:nn4} by $\max\{\norm{x}, \norm{y}\}$ on both sides. Set $\rho := \frac{x^Ty}{\norm{y}}$. We have
	
	\begin{equation*}
		\begin{aligned}
			\sqrt{2}\norm{xx^T - yy^T}_F &\! =\! \sqrt{2}\left(\sum_{i, j}(x_ix_j - y_iy_j)^2\right)^{\frac{1}{2}} \\
			& \! =\! \sqrt{2}\left( \!(\sum_i x_i^2)(\sum_j x_j^2) \!+\! 
		(\sum_i y_i^2)(\sum_j y_j^2) \!-\! 2 (\sum_{i}x_iy_i)(\sum_{j}x_jy_j)\! \right)^{\frac{1}{2}} \\
			& \! =\! \left(2(1+t^4)-4\rho^2 t^2\right)^{\frac{1}{2}}\\
			& \!\geq\! \left((1+t^2)^2-4\rho^2 t^2\right)^{\frac{1}{2}} \\
			& \! =\! \sqrt{1+t^2 + 2\rho t}\sqrt{1+t^2 - 2\rho t} \\
			& \! =\! \norm{x + y}\norm{x - y},
		\end{aligned}
	\end{equation*}
	where the inequality follows by the algebraic geometric mean inequality.
\hfill$\square$

\bibliographystyle{plain}
\bibliography{global_minimizer.bib}
\medskip

\end{document}